\documentclass[12pt]{article}

\usepackage{amssymb,amsfonts,amsthm,fullpage}

\newtheorem{theorem}{Theorem}[section]
\newtheorem{lemma}[theorem]{Lemma}

\newtheorem{cor}[theorem]{Corollary}

\renewcommand{\>}{\rangle}
\renewcommand{\v}{\vec{v}}
\renewcommand{\u}{\vec{u}}

\begin{document}

\title{Recovering the period in Shor's algorithm with Gauss' algorithm for lattice basis reduction}
\author{Allison Koenecke\thanks{Mathematics Department, Massachusetts Institute of Technology, Cambridge, MA, USA; \texttt{allisonk@mit.edu}}\quad\quad Pawel Wocjan\thanks{Mathematics Department \& Center for Theoretical Physics, Massachusetts Institute of Technology, Cambridge, MA, USA; on sabbatical leave from Department of Electrical Engineering and Computer Science, University of Central Florida, Orlando, FL, USA; \texttt{wocjan@eecs.ucf.edu}}}

\date{January 18, 2013}

\maketitle

\begin{abstract}
Shor's algorithm contains a classical post-processing part for which we aim to create an efficient, understandable method aside from continued fractions.

Let $r$ be an unknown positive integer.  Assume that with some constant probability we obtain random positive integers of the form $x=\lfloor N k/r \rceil$ where $\lfloor \cdot \rceil$ is either the floor or ceiling  of the rational number, $k$ is selected uniformly at random from $\{0,1,\ldots,r-1\}$, and $N$ is a parameter that can be chosen. The  problem of recovering $r$ from such samples occurs precisely in the classical post-processing part of Shor's algorithm.  The quantum part (quantum phase estimation) makes it possible to obtain such samples where $r$ is the order of some element $a\in\mathbb{Z}^\times_n$ and $n$ is the number to be factored. 

Shor showed that the continued fraction algorithm can be used to efficiently recover $r$, since if $N>2r^2$ then $k/r$ appears in lowest terms as one of the convergents of $x/N$ due to a standard result on continued fractions.  We present here an alternative method for recovering $r$ based on the Gauss algorithm for lattice basis reduction, allowing us to efficiently find the shortest nonzero vector of a lattice generated by two vectors.  Our method is about as efficient as the method based on continued fractions, yet it is much easier to understand all the details of why it works. 
\end{abstract}

%
%

\section{Introduction}

In the classical post-processing part of Shor's algorithm, the task is to recover the unknown positive integer $r$ from samples of the form 
\begin{equation}\label{eq:form}
\left\lfloor N k/r \right\rfloor \quad\mbox{or}\quad \left\lceil N k/r \right\rceil.
\end{equation}
where $k$ is selected uniformly at random from $\{0,1,\ldots,r-1\}$, and $N$ is chosen to be a sufficiently large power of $2$.  These samples are produced by quantum phase estimation.  

Here $r$ denotes the order of some randomly chosen element $a$ of the unit group $\mathbb{Z}^\times_n$ of the residue ring $\mathbb{Z}_n$ where $n$ is the number to be factored.  The probabilistic reduction of integer factorization to order finding shows that when $a$ is chosen randomly, then there is a high probability that $r$ is even and $\gcd(a^{r/2}-1,n)$ yields a nontrivial factor of $n$.  This is described  in \cite[Subsection 7.3.1]{KLM}.

We now briefly summarize the idea underlying phase estimation.  Let $U$ be the unitary transformation corresponding to the permutation of $\mathbb{Z}_n$ defined by $j\mapsto a\cdot j$, where $a\in\mathbb{Z}_n$.  Observe that orbit of $1$ under this permutation is $1=a^0,a^1,\ldots,a^{r-1}$, implying that $U$ acts as a cyclic shift operator of order $r$ when restricted to the subspace spanned by $|1\rangle, |a\rangle,\ldots,|a^{r-1}\rangle$.  Therefore $|1\>$ is a uniform superposition of eigenvectors $|\psi_k\>$ with eigenvalue $e^{2\pi i k/r}$ for $k=0,\ldots,r-1$. (These correspond to the eigenvectors and eigenvalues of the cyclic shift operator of order $r$ when we identify $U$ with its restriction to the above subspace and use $|j\>$ instead of $|a^{j}\>$ to denote the basis vectors of the subspace.)

The analysis in \cite[Theorem 7.1.5]{KLM} shows that if we run quantum phase estimation of $U$ in the state $|1\>$, then we obtain samples of the form as in (\ref{eq:form}) with 
with probability greater or equal to $\frac{8}{\pi^2}$, where $k$ is selected uniformly at random from $\{0,1,\dots,r-1\}$.  The uniform distribution over $k$ occurs because $|1\>$ is a uniform superposition of the eigenvectors $|\psi_k\>$ and, thus, quantum phase estimation behaves as if we had a classical mixture of the $|\psi_k\>$.
	
Let $x$ be an outcome of the form in as (\ref{eq:form}) implying $\left| \frac{x}{N} - \frac{k}{r} \right| < \frac{1}{N}$.
If $N$ is greater than or equal to $2 r^2$, then $k/r$ in lowest term appears in lowest terms as one of the convergents of the continued fraction expansion of $x/N$.  This is a standard result in the theory of continued fractions (see \cite[Theorem 7.17]{KLM} for a formulation of the results as relevant for the recovery problem and \cite[Chapter 15 and Theorem 15.9]{burton} for a proof).  Note that we need that $\gcd(k,r)=1$ to be able to recover $r$, which happens this probability $\varphi(r)/r=\Omega(1/\log\log r)$.\footnote{This shows that by repeating this experiment only $O(\log \log r)$ times, we are assured of a high probability of success.  See the short discussion of two methods due to Odylzko and Knill, making it possible to achieve constant success probability \cite[page 1501]{Shor}.}

We present here a different method for recovering $r$ that requires two samples and succeeds with constant probability.  Our method relies on the Gauss algorithm for lattice reduction, which makes it possible to efficiently compute shortest (nonzero) lattice vectors in lattices generated by two vectors.\footnote{It can be shown that the two basis vectors returned by the Gauss algorithm are always the two successive minima of the lattice.  But we do not need this more general result.}  There are two mainly pedagogical reasons for considering this alternative method.  First, all the necessary technical details can be explained and proved in a self-contained way in less than five pages, whereas proving that continued fractions yields the desired approximation requires more effort.  Second, this method can be seen as a very special case of a more general method for obtaining an approximate basis of a higher-dimensional lattice $L$ from an approximate generating set of its dual lattice $L^*$, which plays an important rule in quantum algorithms for some number-theoretic problems \cite{FW}.  This higher-dimensional lattice reconstruction problem cannot be solved by methods related to continued fractions.  It is thus helpful to first understand the special case that applies to the simple reconstruction problem arising in Shor's algorithm before studying the the significantly more difficult higher-dimensional problem.  

%
%

\section{Recovering the period with the Gauss algorithm}

Let $x$ and $y$ be two outcomes of the phase estimation algorithm.  Assume that both samples have the form as given in (\ref{eq:form}) and that the corresponding $k$ and $\ell$ are coprime.  It is relatively easy to show that the probability of $k$ and $\ell$ being coprime is greater than $1/2$ (see \cite[Lemma 20]{SW} and \cite{felix} for a better lower bound).

Let $s$ be an integer that we fix later.  Consider the linearly independent vectors
\[
\vec{x} = 
\left(
\begin{array}{c}
1 \\ 
0 \\ 
s x/N
\end{array}
\right)\quad\mbox{and}\quad
\vec{y} = 
\left(
\begin{array}{c}
0 \\ 
1 \\ 
s y/N
\end{array}
\right) \quad\mbox{in $\mathbb{Q}^3$}
\]
and let $L=\mathbb{Z} \vec{x} + \mathbb{Z} \vec{y}$
denote the lattice generated by these two vectors.

At this stage we only need to know that the Gauss algorithm determines two integers $m$ and $n$ such that $\vec{u} = m \vec{x} + n \vec{y}$ is a shortest nonzero vector.  This is proved in \cite[Section 2]{vallee}.  For the sake of completeness, we provide a simplified proof in the next section.

%
%

\begin{theorem}
Let $B$ be an upper bound on the unknown integer $r$.  
Set $s = 4B^2$ and choose an integer $N$ with $N\ge \sqrt{2} s$.
Let $x=\lfloor N k/r \rceil$ and $y=\lfloor N \ell/r \rceil$ be the samples as in (\ref{eq:form}).  Assume that $k$ and $\ell$ are coprime.
Then, the vector $\vec{u} = (-\ell) \vec{x} + k \vec{y}$
is the unique (up to multiplication by $-1$) shortest nonzero vector of $L$.
\end{theorem}
%
%
\begin{proof}
First, consider the vector 
\[
\vec{u} =
(-\ell) \vec{x} + k \vec{y} 
= 
\left(
\begin{array}{c}
-\ell \\ k \\ s \big ( -\ell x/N + k y/N \big)
\end{array}
\right).
\]
The absolute values of the first two entries of $\vec{u}$ are bounded from above by $r-1$ since $k,\ell\le r-1$.
To bound the absolute value of the third entry, we write $x/N = k/r + \xi_x$ and $y/r = \ell / N + \xi_y$ with $|\xi_x|,|\xi_y| \le 1/N$. The triangle inequality implies
\[
|s ( -\ell x/N + k y/N )| =  |s( -\ell k/r + k \ell/r - \ell \xi_x + k \xi_y )| \le 2s(r-1)/N.
\]
We obtain the the upper bound 
\[
\| \vec{u} \|_2 \le \sqrt{ (r-1)^2 + (r-1)^2 + \big( 2s \cdot (r-1)/N \big)^2 } = (r - 1) \sqrt{2 + 4s^2/N^2} \le 2 B
\]
since $k,\ell\le r-1$, $r\le B$, and $N \ge \sqrt{2} s$.

Second, we show that the above vector $\vec{u}$ is the unique (up to multiplication by $-1$) shortest nonzero vector of $L$.  Assume to the contrary that 
$\vec{z}=m \vec{x} + n \vec{y}$ is a shortest nonzero lattice vector with $(m,n) \neq \pm(-\ell,k)$.  Clearly, we must also have $(m,n)\neq c(-\ell,k)$ for all integers $c$ with $|c|\ge 2$ since in this case $\vec{z}=c \vec{u}$ cannot be a shortest nonzero lattice vector.   This implies that $m k + n \ell \neq 0$.

We have
\[
\| \vec{z} \|_2^2 = m^2 + n^2 + \left( m s \frac{x}{N} + n s \frac{y}{N} \right)^2.
\]
We may assume that $\sqrt{m^2 + n^2} \le 2B$ because otherwise $\vec{z}$ would be longer than $\vec{u}$.  We obtain
\begin{eqnarray*}
\| \vec{z} \|_2 
& > &  
\left| m s \frac{x}{N} + n s \frac{y}{N} \right| \\
& = &
s \left| m \left( \frac{k}{r} + \xi_x \right) + n \left( \frac{\ell}{r} + \xi_y \right) \right| \\
& \ge &
\frac{s}{r} \Big| m k + n \ell \Big| - \frac{s}{N} \Big(|m| + |n|\Big) \\
& \ge &
\frac{s}{r} - \frac{s}{N} \sqrt{2} \sqrt{m^2 + n^2} \\
& \ge &
\frac{s}{r} - \frac{s}{N} \sqrt{2} \, 2B  \\
& \ge &
\frac{s}{r} - \frac{s}{2B} \\
& \ge &
\frac{s}{2r} 
\ge 
\frac{s}{2B}  
\ge 
2 B,
\end{eqnarray*}
implying that $\vec{z}$ would be longer than $\vec{u}$.
\end{proof}

The above theorem shows that we can recover the value $k$ corresponding to $x$.  We need the following lemma to show that $Nk/x$ is sufficiently close to the integer $r$.
\begin{lemma}
Let $\zeta,\zeta'\in [a,b] \subset [0,1]$.  Then
\[
\left| \frac{1}{\zeta} - \frac{1}{\zeta'} \right| \le \frac{1}{a^2} |\zeta-\zeta'|.
\]
\end{lemma}
\begin{proof}
This follows since the function $f(\zeta)=1/\zeta$ is Lipschitz continuous with constant given by $\min_{\zeta''\in[a,b]} \{f'(\zeta'')\}=1/a^2$.
\end{proof}

We have 
\[
\left|
\frac{x}{N k } - \frac{1}{r} 
\right|
< \frac{1}{N k} \le \frac{1}{N}.
\]
We now apply the above lemma with $\zeta = 1/r$ and $\zeta' = x/(Nk)$ and $a=1/r - 1/N$ and obtain
\[
\left|
\frac{N k}{x} - r 
\right| \le \frac{1}{a^2 N} < 1. 
\]

We now see that we have to choose $N$ on the order of $B^2$ to obtain an estimate that is close to $r$.

%
%

\section{Gauss algorithm}

Let $\vec{u}$ and $\vec{v}$ be two arbitrary vectors in $\mathbb{Z}^d$ and $M:=\max\{\|\vec{u}\|,\|\vec{v}\|\}$.  We refer to $M$ as the length of the basis $\vec{u},\vec{v}$.  We show that the Gauss algorithm makes it possible to determine a shortest nonzero vector of the lattice $\mathbb{Z}\vec{u}+\mathbb{Z}\vec{v}$ in time that scales polynomially in $d$ and $\log(M)$.  We summarize and simplify the necessary results in \cite[Section 2]{vallee}.  

To apply the Gauss algorithm to the vectors $\vec{x}$ and $\vec{y}$ from the previous section, we have to multiply them by $N$ to ensure that all entries are integers.  The parameter $d$ is equal to $3$ in this case.

We need two definitions to present and analyze the algorithm.  
For $f\in\mathbb{Q}$, define the closest integer to $f$ to be the unique integer $m$ such that $f-m\in(-\frac{1}{2},\frac{1}{2}]$.  We denote the closest integer to $f$ by $[f]$.  For $f\in\mathbb{Q}$, define the sign of $f$ to be $+1$ if $f$ is nonnegative and $-1$ otherwise.  We denote the sign of $f$ by $s(f)$.

We start with the basis $\vec{u}$ and $\vec{v}$ where we assume that $\|\vec{u}\| \le \|\vec{v}\|$.  
We replace the vector $\vec{v}$ by the shortest vector $\chi(\vec{v},\vec{u})$ of the set
\[
K(\v,\u) := \{ \vec{w} \mid \vec{w} = \varepsilon (\v - m \u), m\in\mathbb{Z}, \varepsilon = \pm 1\}
\]
that makes an acute angle with $\u$.  Note that $\chi(\v,\u)$ is easy to calculate from $f=\u \cdot \v / \| \u \|^2$.  

We see that $f$ is the solution to the quadratic minimization problem $\|\vec{w}(f)\|^2 = \|\v - f \u\|^2 $ with respect to $f$.  This yields the optimal value of the concave-up parabola to be $f$ where $f\in\mathbb{R}$.  However, if we are required to use integer values, we have that $[f] \in\mathbb{Z}$ gives us the shortest norm.
Hence the optimal integer $m$ is equal to $[f]$ and $\varepsilon$ is is equal to $s(f-[f])$.

\begin{itemize}
\item[] REPEAT
\begin{itemize}
\item[1.] IF $\| \u \|^2 > \| \v \|^2$, exchange $\u$ and $\v$;
\item[2.] $\v := \chi( \v, \u )$;
\end{itemize}
\item[] UNTIL $\| \u \|^2 \le \| \v \|^2$. 
\end{itemize}

The following result describes the output configuration:

%
%

\begin{lemma}[Shortest lattice vector]
Given an arbitrary basis $\u,\v$ of a lattice $L$  in $\mathbb{Z}^d$, the Gauss algorithm outputs a shortest nonzero vector of $L$.\footnote{It can be shown that the two vectors output by the Gauss algorithm are the two successive minima of $L$.  But we do not need this stronger result.}  
\end{lemma}
\begin{proof}
The output configuration $\u,\v$ satisfies the two conditions
\[
\| \v \|^2 \ge | \u \|^2 \quad \mbox{and} \quad 0\le\u \cdot \v \le \frac{1}{2} \|\u\|^2.
\]
The first condition corresponds directly to the criterion in the UNTIL statement.  The second condition is seen as follows. By definition of the vector $\chi(\v,\u)$ in step 2 we have
\[
0 \le \u \cdot \chi(\v, \u) =
\varepsilon 
\left(
\frac{\u \cdot \v}{\| \u \|^2} - 
\left[ 
\frac{\u\cdot\v}{\|u\|^2} 
\right]
\right) \| \u \|^2.
\]
Clearly, the absolute value of the term in the round parenthesis is at most $\frac{1}{2}$, which implies the second condition.

We now show that the length of the projection of $\v$ orthogonally to $\u$ is greater than $\frac{\sqrt{3}}{2}\| \u \|$.  Express $\v$ as  
\[
\v = \frac{\u\cdot\v}{\|u\|^2} \, \u + \vec{t},
\]
where $\vec{t}$ is orthogonal to $\u$.  Then $\|\v\|^2 \le \frac{1}{4} \|\u\|^2 + \| \vec{t} \|^2$ since the scalar in front of $\u$ in the above expression is in $[0,\frac{1}{2}]$.  Because $\|\u\|^2 \leq \|\v\|^2$, we have $\|\vec{t}\| \ge \frac{\sqrt{3}}{2}\|\u\|$ as claimed.

We are now ready to show that $\u$ is a shortest lattice vector.  Consider vectors of the form
\[
\beta \vec{v} + \alpha \vec{u} \quad \mbox{with $\beta=\pm 1$ and $\alpha\in\mathbb{Z}$.}
\]
Any such vector has length at least $\|\v\|\ge\|\u\|$ due to the choice of $\v := \chi(\v,\u)$ in step 2.  Recall that the parameter $m$ is always chosen so that the length of the resulting vector $\chi(\v,\u)$ is minimal.  Hence any subsequent addition of an integer multiple of $\vec{u}$ to $\chi(\v,\u)$ cannot decrease the length.

Consider vectors of the form 
\[
\beta \v + \alpha \u \quad \mbox{with $|\beta|\ge 2$ and $\alpha\in\mathbb{Z}$.}
\]
Any such vector has length at least $|\beta| \|t\| \ge 2 \,\frac{\sqrt{3}}{2}\|\u\| > \|\u\|$.

The only vectors not covered by the previous two cases are multiples of $\u$.
\end{proof}

To analyze the complexity of the Gauss algorithm we now describe a modified algorithm that depends on a parameter $t$, which is strictly greater than $1$. The new algorithm is called the Gauss$(t)$ algorithm, and is equivalent to the Gauss algorithm for $t=1$. In Gauss$(t)$, the original loop termination condition 
$\| \u \|^2 \le \| \v \|^2 $
is replaced by
\begin{equation}\label{eq:condT}
\| \u \|^2 \le t^2 \, \| \v \|^2.
\end{equation}
The polynomial time complexity of Gauss$(t)$ is clear.  In each loop, the length of the longer vector is decreased by a factor of at least $1/t$ and the length of any nonzero vector of $L\subseteq\mathbb{Z}^d$ is at least $1$.  We obtain an upper bound on the number $k_t$ of iterations of this algorithm executed on a basis of length $M$:
\[
k_t \le \lceil\log_t(M)\rceil.
\]

%
%

\begin{lemma}
Let $k$ and $k_t$ denote the number of iterations of the Gauss algorithm and the Gauss$(t)$ algorithm when applied to the same basis $\u,\v$ of a lattice $L$.
For any $t\le\sqrt{3}$, the two numbers $k_t$ and $k$ satisfy
\[
k_t \le k \le k_t + 1.
\]
\end{lemma}
\begin{proof}
The inequality $k_t\le k$ is clear.  To prove the upper bound consider the last loop of the Gauss$(t)$ algorithm.  Its output configuration satisfies the two conditions
\[
0 \le \u \cdot \v \le \frac{1}{2} \| \u \|^2 \quad\mbox{and}\quad \|\u\|^2 \le t^2 \|\v\|^2.
\]
If $\| \u \|^2 \le \| \v \|^2$ holds, then this is also the last loop of the Gauss algorithm.  So assume that $\| \u \|^2 > \| \v \|^2$. In this case the Gauss algorithm proceeds by exchanging $\u$ and $\v$.  We denote the configuration after this step by $\vec{u'}=\v$ and $\vec{v'}=\u$.  

We have $\|\vec{u'}\|^2 < \|\vec{v'}\|^2$ and
\[
0 \le \frac{\vec{u'}\cdot\vec{v'}}{\| \vec{u'} \|^2} = \frac{\vec{u'}\cdot\vec{v'}}{\| \vec{v'} \|^2} \cdot \frac{\|\vec{v'}\|^2}{\|\vec{u'}\|^2} =
\frac{\vec{v}\cdot\vec{u}}{\| \vec{u} \|^2} \cdot \frac{\|\vec{u}\|^2}{\|\vec{v}\|^2} \le \frac{1}{2} t^2 \le \frac{3}{2}.
\]
The second inequality implies that there are only two cases we need to consider for the new vector $\chi(\vec{v'},\vec{u'})$ in step two of the Gauss algorithm, which are either $\vec{v'}$ or $\pm (\vec{v'} - \vec{u'})$.  If the first case, the Gauss algorithm stops because $\vec{u'}$ is still shorter than the new vector.  In the second case, we have
\[
\| \vec{v'} - \vec{u'} \| = \| \u - \v \| = \| \v - \u \| \ge \|\v\| = \| \vec{u'} \|.
\]
The inequality is due the particular choice of the vector $\v$ in step two of the Gauss$(t)$ algorithm.  Recall that any subsequent addition of a multiple of $\u$ cannot decrease its length.  Hence the Gauss algorithm also terminates in the second case.

\end{proof}

\begin{cor}
The number of iterations $k$ of the Gauss algorithm executed on a basis of length $M$ satisfies
\[
k \le \lceil\log_{\sqrt{3}}(M)\rceil + 1.
\]
\end{cor}

\subsection*{Acknowledgments}
We would like to thank Nolan Wallach for helpful discussions. 

A.K. gratefully acknowledges the support from MIT's Undergraduate Research Opportunties Program (UROP) under Peter Shor's supervision.
P.W. gratefully acknowledges the support from the National Science Foundation CAREER Award CCF-0746600.  This work was supported in part by the National Science Foundation
Science and Technology Center for Science of Information, under
grant CCF-0939370.

\bibliographystyle{amsplain}

\begin{thebibliography}{10}

\bibitem{KLM}
P.~Kaye, R.~Laflamme, and M.~Mosca, \emph{An Introduction to Quantum Computing}, Oxford, 2007.

\bibitem{burton}
D.~Burton, \emph{Elementary Number Theory}, McGraw-Hill, 7th edition, 2010.

\bibitem{felix}
F.~Fontein, ``The Probability that Two Numbers Are Coprime'' on Felix' Math Page;
{\small \texttt{http://math.fontein.de/2012/07/10/the-probability-that-two-numbers-are-coprime/}}

\bibitem{FW}
F.~Fontein and P.~Wocjan, \emph{Quantum Algorithm for Computing the Period Lattice of an Infrastructure}, arXiv preprint, 2011;
\texttt{http://arxiv.org/abs/1111.1348}

\bibitem{Shor}
P.~Shor, \emph{Polynomial-Time Algorithms for Prime Factorization and Discrete Logarithms on a Quantum Computer}, 
SIAM J. Comput., 26(5), 1484Ð1509, 1997.

\bibitem{SW}
P.~Sarvepalli and P.~Wocjan, \emph{Quantum Algorithms for One-Dimensional Infrastructures}, arXiv preprint 2011;  
\texttt{http://arxiv.org/abs/1106.6347}

\bibitem{vallee}
B.~Vall\'ee, \emph{Gauss Algorithm Revisted}, Journal of Algorithms, 12, pp.~556--572, 1991.

\end{thebibliography}

\end{document}